\documentclass[letter,11pt]{article}
\usepackage[letterpaper,text={6.5in,9in}]{geometry}

\usepackage{authblk}
\usepackage{hyperref}
\usepackage{amsfonts}
\usepackage{amsmath}
\usepackage{amssymb}
\usepackage{amsthm}
\usepackage{graphicx}
\usepackage{color}
\usepackage[none]{hyphenat}
\usepackage{wrapfig}
\usepackage[justification=centering]{caption}
\usepackage{bm}
\usepackage{enumitem}
\usepackage[T1]{fontenc}
\usepackage[ansinew]{inputenc}

\theoremstyle{plain}
\newtheorem{theorem}{Theorem}[section]
\newtheorem{lemma}[theorem]{Lemma}

\newtheorem{corollary}[theorem]{Corollary}
\theoremstyle{definition}
\newtheorem{definition}[theorem]{Definition}

\title{Consensus in Equilibrium: Can One Against All Decide Fairly?}

\usepackage{algorithm}
\usepackage[]{algpseudocode}
\theoremstyle{plain}
\theoremstyle{definition}
\setlist[enumerate]{itemsep=0mm,topsep=1.2mm}
\setlist[itemize]{itemsep=0mm,topsep=1.2mm}

\newcommand{\EDP}{RIS} %Equivalent distributed problem

\begin{document}
	
\begin{titlepage}
	\title{Consensus in Equilibrium: \\
		Can One Against All Decide Fairly?
		{\let\thefootnote\relax\footnote{
			This work is partially supported by a grant from the Blavatnik Cyber Security Council and
		the Blavatnik Family Computer Science Research Fund.}}
	}
	
	\author{Yehuda Afek}
	\author{Itay Harel}
	\author{Amit Jacob-Fanani}
	\author{Moshe Sulamy}
	\affil{Tel-Aviv University}
	\date{}
	\maketitle{}

	\begin{abstract}
		
Is there an equilibrium for distributed consensus when all agents except one collude to steer the decision value towards their preference?
If an equilibrium exists, then an $n-1$ size coalition cannot do better by deviating from the algorithm,
even if it prefers a different decision value.
We show that an equilibrium exists under this condition only if the number of agents in the network is odd
and the decision is binary (among two possible input values).
That is, in this framework we provide a separation between binary and multi-valued consensus.
Moreover, the input and output distribution must be uniform,
regardless of the communication model (synchronous or asynchronous).
Furthermore, we define a new problem - Resilient Input Sharing (RIS),
and use it to find an {\em iff} condition
for the $(n-1)$-resilient equilibrium for deterministic binary consensus,
essentially showing that an equilibrium for deterministic consensus
is equivalent to each agent learning all the other inputs in some strong sense.
Finally, we note that $(n-2)$-resilient equilibrium for binary consensus is possible for any $n$.
The case of $(n-2)$-resilient equilibrium for \emph{multi-valued} consensus is left open.

	\end{abstract}

\end{titlepage}
\pagenumbering{arabic}

\section{Introduction}
\label{section_intro}

In recent years, there is a growing interest in distributed algorithms for networks of rational agents
that may deviate from the prescribed algorithm in order to increase their profit
\cite{DisMeetsGame,DistRobust,LeaderADH,BARFault,RationalSecret}. 
For example, an agent may have a higher profit if zero is decided in a consensus algorithm,
or an agent may prefer to be (or not to be) the elected leader in a leader election algorithm. 
The goal is to design distributed algorithms that reach \emph{equilibrium}, that is,
where no agent can profit by cheating.

In this paper we study the consensus problem in a network of rational agents,
in which each agent has a preferred decision value. We consider $(n-1)$-resilient equilibrium, that is,
an equilibrium that is resilient to any coalition of up to $n-1$ agents
that may collude in order to increase their expected profit (utility).
This problem was proposed in \cite{LeaderADH} and studied also in \cite{GTBlocks},
where the authors suggest an $(n-1)$-resilient equilibrium for binary consensus
in a synchronous ring.

We prove that in any $(n-1)$-resilient equilibrium for binary consensus,
the output of the agents must be the XOR of the inputs of all agents.
Thus, due to validity, there is \emph{no} $(n-1)$-resilient equilibrium for binary consensus
in {\em even} sized networks,
and the algorithm in \cite{GTBlocks} works well only for odd sized networks.
Still, we show that the algorithm in \cite{GTBlocks} reaches $(n-2)$-resilient equilibrium
for binary consensus with uniform input distribution, for \emph{any} $n$.

We further show that multi-valued consensus is impossible,
i.e., there is no $(n-1)$-resilient equilibrium for multi-valued consensus
for $r>2$ where $r$ is the number of possible values,
thus surprisingly there is a computational gap between binary and multi-valued consensus in this model. Note that it was previously shown that in this game theoretic model,
leader election is also not equivalent to consensus \cite{GTBlocks}.
 
Furthermore, we show that in this model, {\em deterministic} binary consensus
is equivalent to resilient input sharing (RIS), a natural problem in distributed computing
in which each agent $i$ shares its input with
all other agents in the network
(a variant of the knowledge sharing problem defined in \cite{GTBlocks}). 
That is, in any odd sized network with uniform input distribution,
any algorithm for RIS can be transformed into
a $(n-1)$-resilient equilibrium for deterministic binary consensus and vice versa.
Thus, providing a sufficient and necessary condition
for $(n-1)$-resilient equilibrium for deterministic binary consensus.

%%%%%%%%%%%%%%%%%%%%%%%%%%%%%%%%%%%%%%%%%%%%%%

\subsection{Our Contributions}
\label{sub_contrib}
%As opposed to the secret sharing problem studied by Abraham et al.\cite{DistRobust},
%we assume that the goal of the coalition is to alter the result of the consensus,
%not to learn the secret faster than agents outside of the coalition.
are as follows:
\begin{description}
	\item[$(\S \ref{section_xor})$] Any $(n-1)$-resilient equilibrium for binary consensus
	decides on the XOR of all input values.
	
	\item[$(\S \ref{section_xor_lemmas})$] In any $(n-1)$-resilient equilibrium for binary consensus
	the input and output distributions are uniform.
	
	\item[$(\S \ref{section_cover_bin_cons})$] The protocol suggested in \cite{GTBlocks}
	reaches $(n-2)$-resilient equilibrium for binary consensus with uniform input distribution,
	for any $n$.
	
	\item[$(\S \ref{section_multi})$]
	There is no $(n-1)$-resilient equilibrium for multi-valued consensus for $r>2$ possible inputs.
	
	\item[$(\S \ref{section_sufficient})$] \emph{Deterministic} $(n-1)$-resilient equilibrium for
binary consensus in a network exists \emph{iff}:
\begin{enumerate}
	\item The network size is odd.
	\item The input distribution is uniform.
	\item An equilibrium for Resilient Input Sharing (RIS) is possible in the network topology.
\end{enumerate}

%	\item From the proof that the ability to perform consensus implies ability to perform \EDP{}, we conclude that in order to implement consensus, the network unintentionally implements knowledge sharing 

%	\item There is no $n-1$-strong equilibrium for consensus
%	if the input distribution is non-uniform,
%	i.e., each agent must have a probability of $\frac{1}{2}$ to get $0$ as input
%	and a probability of $\frac{1}{2}$ to get $1$ as input.

%	\item We define a new distributed problem – the \emph{Resilient Knowledge Sharing (RKS)} problem.
%	In this problem, the processors need to perform a regular knowledge sharing,
%	where each processor $i$ has some networkinformation it needs t-K-o share will all other processors;
%	In RKS, the sharing is done under a new condition -– an agent $i$ must not receive messages
%%	from an agent $j$ at time $t$, if by time $t$, agent $j$ could have learned $i$'s information. 

\end{description}
The model, notations and some definitions are given in Section \ref{section_model}, and we discuss our results and further thoughts in Section \ref{section_discussion}.

\subsection{Related Work}

The secret sharing problem \cite{HowToShare} initiated the connection between
distributed computing and game theory.
Further works in this line of research considered multiparty communication
with Byzantine and rational agents
\cite{DisMeetsGame,ScalableRational,EfficientRational,ByzantineWithRational,RationalityAndAdversarial}.

In \cite{LeaderADH}, the first distributed protocols for a network of rational agents
are presented, specifically protocols for \emph{fair} leader election.
In \cite{GTBlocks}, the authors continue this line of research by providing
basic building blocks for game theoretic distributed algorithms,
namely a wake-up and knowledge sharing building blocks that are in equilibrium,
and equilibria for consensus, renaming, and leader election are presented using these building blocks.
The consensus algorithm in \cite{GTBlocks} claims to reach $(n-1)$-resilient equilibrium in
a ring or complete network, using the knowledge sharing building block to share the input
of all processors in the network, and outputting the XOR of all inputs.
Consensus was further researched in \cite{RationalConsensus}, where the authors show that
there is no ex-post Nash equilibrium for rational consensus,
and present a Nash equilibrium that tolerates $f$ failures under some minimal assumptions
on the failure pattern.
Equilibrium for fair leader election and fair coin toss are also presented and discussed in \cite{YifrachLeader},
where it is shown to be resilient only to coalitions of sub-linear size,
and a modification to the leader election protocol from \cite{LeaderADH,GTBlocks}
that is resilient to every coalition of size $\Theta(\sqrt{n})$ is proposed.

In \cite{GTColor}, the authors examine the impact of a-priori knowledge of the network size
on the equilibrium of distributed algorithms, assuming the $id$ space is unlimited
and thus vulnerable to a Sybil attack \cite{SybilAttack}.
In \cite{GTidspace} the authors remove this assumption and assume the $id$ space is bounded,
examining the relation between the size of the $id$ space and the number of agents in the network
in which an equilibrium is possible.

\section{Model}
\label{section_model}

We use the standard message-passing model,
where the network is a bidirectional graph $G=(V,E)$ with $|V|=n$ nodes,
each node representing a \emph{rational} agent, following the model in \cite{DistRobust,LeaderADH}.
We assume $n$ is a-priori known to all agents,
$G$ is $2$-vertex-connected,
and all agents start the protocol together, i.e., all agents wake-up at the same time.
We can use the Wake-Up \cite{GTBlocks} building block to relax this assumption.
In Sections~\ref{section_necessary} and~\ref{section_multi} the results apply
for both synchronous and asynchronous communication networks,
while Section~\ref{section_sufficient} assumes a synchronous network.

In the consensus problem, each agent $i$ has an id $id_i$ and an input $I_i \in \{0,...r-1\}$
and must output a decision $D_i \in \{0,...r-1, \bot\}$.
The $\bot$ output can be output by an agent to abort the protocol
when a deviation by another agent is detected.
A protocol achieves consensus if it satisfies the following \cite{Cons4}:
\begin{itemize}
	\item \textbf{Agreement}: All agents decide on the same value, $\forall{i,j}: D_i=D_j$.
	\item \textbf{Validity}: If $v$ was decided then it was the input of some agent,
		$\forall j \exists i: D_j = I_i$.
	\item \textbf{Termination}: Every agent eventually decides, $\forall{i}: D_i \neq \bot$.
\end{itemize}

\begin{definition}[Protocol Outcome]
		The outcome of the protocol is determined by the input and output of all agents.
		An outcome is \emph{legal} if it satisfies agreement, validity, and termination,
		otherwise the outcome is \emph{erroneous}.
\end{definition}

Considering individual rational agents, each agent $i$ has a utility function $U_i$ over the possible outcomes of the protocol.
The higher the value assigned by $U_i$ to an outcome,
the better this outcome is for $i$.
We assume the utility function $U_i$ of each agent $i$ satisfies \emph{Solution Preference} \cite{LeaderADH}:
%which guarantees that an agent never has an incentive to fail the algorithm.

\begin{definition}[Solution Preference]
	The utility function $U_i$ of any agent $i$ never assigns a higher utility to
	an erroneous outcome than to a legal one.
\end{definition}

Thus, the Solution Preference guarantees that an agent never has an incentive to sabotage the protocol,
that is, to prefer an outcome that falsifies either agreement or validity, or termination.
However, agents may take risks that might lead to erroneous outcomes
if these risks also lead to a legal outcome which increases their expected utility,
that is, if these risks increase the expected utility that the agent is expected to gain.

An intuitive example for a utility function of an agent $I$
with a preference towards a decision value of $1$ is:
$$
U_i = 
\begin{cases}
\text{1} &\quad \exists j: I_j=1 \land \forall k: D_k=1 \text{ ($1$ is decided by all agents)} \\
\text{0} &\quad $otherwise$ \text{ ($0$ is decided or erroneous outcome)} \\
\end{cases}
$$

All agents are given a protocol at the start of the execution,
but any agent may deviate and execute a different protocol if it increases its expected utility.
A protocol is said to \emph{reach equilibrium} if no agent can unilaterally increase its expected utility
by deviating from the protocol.

\begin{definition}[Nash Equilibrium\footnotemark]
	A protocol $\Phi$ is said to reach equilibrium if, for any agent $i$,
	there is no protocol $\Psi \neq \Phi$ that $i$ may execute and leads
	to a higher expected utility for $i$, assuming all other agents follow $\Phi$.
\end{definition}

\footnotetext{
		Previous works defined equilibrium over each step of the protocol.
		For convenience, this definition is slightly different,
		but it is easy to see that it is equivalent.}

\subsection{Coalitions}
We define a coalition of size $t$ as a set of $t$ rational agents that
cooperate to increase the utility of each agent in $t$.
A protocol that reaches $t$-resilient equilibrium \cite{LeaderADH} is resilient to coalitions of size up to $t$,
that is, no group of $t$ agents or less has an incentive to collude and deviate from the protocol.
We assume coalition members may agree on a deviation from the protocol in advance,
but can communicate only over the network links during the protocol execution.

\begin{definition}[$t$-resilient Equilibrium]
	A protocol $\Phi$ is said to reach $t$-resilient equilibrium if,
	for any group of agents $C \subset V$ s.t., $|C| \leq t$,
	there is no protocol $\Psi (\neq\Phi)$ that agents in $C$ may execute
	and which would lead to a higher expected utility for each agent in $C$,
	assuming all agents not in $C$ follow $\Phi$.
\end{definition}

The same intuitive example for a utility function above holds for a coalition,
in which the coalition has a preference towards a decision value $1$.

\subsection{Notations}
The following notations are used throughout this paper:
\begin{itemize}	
	\item $S_{-i}$ - all possible input vectors of agents in $V\setminus\{i\}$.
	\item $\#(b)$ - the number of agents in $V$ that receive $b$ as input.
	\item $\#_{-i}(b)$ - the number of agents in $V\setminus\{i\}$ that receive $b$	as input.	
%	\item XOR of binary set of values $X$ - $\bigoplus\limits_{v \in X} v$ = $(\sum\limits_{v \in X} v) \mod 2$.
	\item $I_i$ - the input of agent $i$.
	\item $D_i$ - the output value decided by agent $i$ at the end of the algorithm.
	\item $r$ - the number of possible input and output values. For binary consensus: $r=2$.
\end{itemize}

\section{Necessary Conditions for ($n-1$)-resilient Consensus}
\label{section_necessary}
% k >> l, m >> k, r >> m

\begin{theorem}
	\label{th_xor}
	The decision of any $(n-1)$-resilient equilibrium for binary consensus
	must be the XOR of all inputs, that is,
	$\forall{i}: D_i = \bigoplus\limits_{j \in V} I_j = \sum\limits_{j \in V} I_j \mod 2$
\end{theorem}

Before we turn to the proof of Theorem~\ref{th_xor}
given in sections \ref{section_xor} and \ref{section_xor_lemmas},
note that according to this theorem,
if $n$ is even and all inputs are $1$ the decision must be $0$,
contradicting validity and leading to the following corollary:

\begin{corollary}
	\label{cor_noeven}
	There is no $(n-1)$-resilient equilibrium for binary consensus for even sized networks %when  $n$, the network size, is even
\end{corollary}

\subsection{Output is the XOR of the Inputs}
\label{section_xor}
Here we prove Theorem~\ref{th_xor} based on the following two theorems, that are proved in Section~\ref{section_xor_lemmas}:

\begin{theorem}
	\label{th_uniform}
	If the distribution over the inputs is not uniform,
	there is no $(n-1)$-resilient equilibrium for consensus, i.e.:
	$\forall v_1,v_2: P[I_i=v_1]=P[I_i=v_2] = \frac{1}{r}$
\end{theorem}

\begin{theorem}
	\label{th_uniform_out}
	In any $(n-1)$-resilient equilibrium for consensus, given any $n-1$ inputs,
		the distribution over the possible decision values is uniform:
	$\forall s \in S_{-i},v \in \{0,...r-1\}: P[D_i=v|s] = \frac{1}{r}$
\end{theorem}

Notice that while the proof of theorem~\ref{th_xor} holds only for \emph{binary} consensus, theorems~\ref{th_uniform} and~\ref{th_uniform_out} are correct for multi-valued consensus as well.

\begin{proof}[Proof of Theorem~\ref{th_xor}]
	We prove that the decision value of binary consensus must be the XOR of all inputs
	using induction on $\#(1)$,
	the number of agents in the network whose input value is $1$.
	
	In the base case $\#(1)=0$, the input of all agents is $0$.
	By validity the decision must be $0$.
	
	For clarity of exposition we spell out the next case of the induction, $\#(1)=1$, i.e.,
	the input of one agent is $1$ and of all other $n-1$ agents is $0$.
	Assume by contradiction that the probability that $0$ is decided in this case is greater than $0$, i.e.,
	$$\exists i\in V:\ P[D_i = 0\ |\ \#_{-i}(0)=0\ \land\ I_i = 1] = p > 0 $$
	Let $s_{0}$ be an input configuration for a coalition in which all members of the coalition
	(i.e., $V\setminus \{i\}$) claim to receive 0 as input, i.e., $\#_{-i}(1)=0$.
	Notice that:
	\begin{align}
	P[D_i = 0 |\ s_0] &= P[I_i=0]\cdot P[D_i = 0\ |\ s_0 \land I_i = 0] + P[I_i=1]\cdot P[D_i = 0\ |\ s_0 \land I_i = 1] \notag\\
	&= P[I_i=0]\cdot P[D_i = 0\ |\ base\ case\ ] + P[I_i=1]\cdot P[D_i = 0\ |\ s_0 \land I_i = 1] \notag\\
	&= P[I_i=0]\cdot 1 + P[I_i=1]\cdot p \notag
	\end{align}
	By Theorem~\ref{th_uniform} (and since this is binary consensus) it follows that:
	$$P[D_i = 0 |\ s_0] =  \frac{1}{2} \cdot 1 + \frac{1}{2} \cdot p > \frac{1}{2}$$
	Thus, contradicting Theorem~\ref{th_uniform_out} and proving that, $\forall i\in V$:
	$ P[D_i = 0\ |\ \#_{-i}(0)=0\ \land\ I_i = 1] = 0$
	Thus if $\#(1)=1$, the decision value must be $1$, proving the first induction step. 

	By the inductive assumption, $\forall \#(1) < m$
	the decision value of the consensus must be the XOR of all inputs, i.e.,
	$\#(1)\ mod\ 2$.
	Let $s_{m-1}$ be an input configuration for the coalition ($V\setminus \{i\}$) in which $\#_{-i}(1) = m-1$,
	that is, $m-1$ members of the coalition claim to receive $1$, and the rest $0$.
	
	From Theorem~\ref{th_uniform_out} (and since this is binary consensus) we get: 
	$$P[D_i = (m\ mod\ 2)\ |\ s_{m-1}] = \frac{1}{2}$$
	If $I_i=0$ (which from Theorem~\ref{th_uniform} happens with probability $\frac{1}{2}$) and the coalition acts as if its input is $s_{m-1}$, then $\#(1) = m-1$.
	By the induction hypothesis,
	in such a case the decision value of the consensus must be $m-1\ mod\ 2$.
	To satisfy the equation above it must hold that:
	$$P[D_i = (m\ mod\ 2)\ |\ s_{m-1}\land I_i=1] = 1$$
	Hence, in case $\#(1) = m$, the decision value must be $m\ mod\ 2$ - the XOR of all inputs.
\end{proof}
	
\subsection{Proving Theorems~\ref{th_uniform} and~\ref{th_uniform_out}}
\label{section_xor_lemmas}

While the above proof holds only for \emph{binary} consensus,
the following lemmas and theorems are correct for multi-valued consensus.

\begin{lemma}
	\label{lemma_set_equal}
	In any $(n-1)$-resilient equilibrium for consensus,
	for any $v \in \{0,\dots,r-1\}$,
	given any $n-1$ inputs,
	the probability to decide $v$ is the same:\\
	$$\forall i \in V, s_1,s_2 \in S_{-i}, v: P[D_i=v|s_1] = P[D_i=v|s_2]$$
\end{lemma}
\begin{proof}
	Assume by contradiction that
	$\exists i \in V, s_1,s_2 \in S_{-i},v: P[D_i=v|s_1] < P[D_i=v|s_2]$.
	A coalition $C = V \setminus \{i\}$ with a preference to decide $v$,
	and that receives $s_1$ as input, has an incentive to deviate and act as if their input is $s_2$,
	contradicting equilibrium.
\end{proof}

\begin{lemma}
	\label{lemma_out_in}
	In any $(n-1)$-resilient equilibrium for consensus, for any input $v\in \{0,\dots,r-1\}$,
	the probability to decide $v$ is the same as the probability to receive $v$ as an input:
	$$\forall i \in V, s \in S_{-i}, v: P[D_i=v | s] = P[I_i=v]$$
\end{lemma}

\begin{proof}
	For any $v$, if all inputs are $v$ then by validity $v$ is decided.
	For any agent $i$, let $\tilde{s}=(v,\dots,v) \in S_{-i}$,
	then due to validity, the probability that $v$ is decided is at least $P[I_i=v]$,
	i.e., $P[D_i=v|\tilde{s}] \geq P[I_i=v]$.
	By Lemma~\ref{lemma_set_equal} this is true for any $s \in S_{-i}$.
	Thus, $P[D_i=v] \geq P[I_i=v]$.
	Since $\sum\limits_v P[D_i=v]=1$ and $\sum\limits_v P[I_i=v]=1$,
	then: $\forall s \in S_{-i},v: P[D_i=v|s] = P[I_i=v]$.
\end{proof}

\begin{proof}[Proof of Theorem~\ref{th_uniform}]
	Assume by contradiction that $\exists v1,v2: P[I_i=v_1] > P[I_i=v_2]$.
	
	If all agents receive as input the same value $v_1$, then by validity $v_1$ is decided.
	Given $s=(v_1,\dots,v_1)\in S_{-i}$,
	the probability that $v_1$ is decided is at least the probability that the input of agent $i$ is $v_1$,
	i.e., $P[D_i=v_1 | s] \geq P[I_i=v_1]$.
	
	If $n-1$ agents receive $v_1$ as input and one agent receives $v_2 \neq v_1$ as input
	the decision must not be $v_1$
	otherwise $P[D_i=v_1 | s] > P[I_i=v_1]$ contradicting Lemma~\ref{lemma_out_in},
	thus due to validity the decision must be $v_2$
	when $n-1$ agents receive $v_1$ and one agent receives $v_2$.
	%	$P[D_i=v_2 | s] \geq P[I_i = v_2]$.
	
	Let $s' = (v_2, v_1, \dots, v_1) \in S_{-i}$.
	If agent $i$ receives $v_1$ as input then as stated above $v_2$ is decided, thus:
	$P[D_i=v_2 | s'] \geq P[I_i=v_1] > P[I_i=v_2]$,
	contradicting Lemma~\ref{lemma_out_in}.
	
	Thus, the input distribution must be uniform, i.e.:
	$\forall v_1,v_2: P[I_i=v_1]=P[I_i=v_2]=\frac{1}{r}$.
\end{proof}

\begin{proof}[Proof of Theorem~\ref{th_uniform_out}]
	Combining Lemma~\ref{lemma_out_in} with Theorem~\ref{th_uniform} :
	$$\forall s \in S_{-i},v \in \{0,...r-1\}: P[D_i=v|s] = P[I_i=v]=\frac{1}{r} ~\qedhere$$
\end{proof}

\subsubsection{$(n-2)$-resilient Binary Consensus for any $n$}
\label{section_cover_bin_cons}
A binary consensus protocol for any $n$ is presented in \cite{GTBlocks}
combining a leader election algorithm with a XOR on selected inputs.
In Appendix~\ref{section_res_even} we prove that this protocol reaches $(n-2)$-resilient equilibrium
for binary consensus for any $n$, when the input distribution is uniform.
Note that the algorithm in \cite{GTBlocks} does not work in any network topology,
but on any network in which Resilient Input Sharing is possible
(see \cite{GTBlocks} and Section~\ref{section_sufficient}).

%(leader election, then XOR of all inputs omitting the input of the leader)
%is suggested for these networks in \cite{GTBlocks}.
%In Appendix~\ref{section_res_even} we present said protocol and prove it to reach $(n-2)$-resilient equilibrium
%for networks with uniform input distribution.

\section{No $(n-1)$-resilient Equilibrium for Multi Valued Consensus}
\label{section_multi}

Here we discuss multi-valued consensus,
where the agreement is between $r>2$ possible values rather than two values.
Applying the same logic as in the proof of Theorem~\ref{th_xor} one can deduce:
\begin{lemma}
	\label{lemma_multi_val_helper}
	\quad
	\begin{enumerate}
		\item $\forall i\in V, v\in \{0,\dots,r-1\}: P[D_i=v | \#(0) = n-1 \land \#(v) = 1] = 1$
		\item $\forall i\in V, v\in \{0,\dots,r-1\}: P[D_i=0 | \#(0) = n-2 \land \#(v) = 2] = 1$
	\end{enumerate}
\end{lemma}

\begin{proof}
The proof is the same as the first and second induction steps in the proof of Theorem~\ref{th_xor}.
\end{proof}

\begin{theorem}
	\label{th_multi}
	There is no $(n-1)$-resilient equilibrium for multi-valued consensus for any $r>2$
\end{theorem}

\begin{proof}
	Assume towards a contradiction that there is an $(n-1)$-resilient equilibrium for multi-valued consensus for some $r>2$.
	Let $v,u \in \{1,\ldots, r-1\}$ s.t. $v \neq u$. 
	Denote by $X$ any configuration in which the input of one agent is $v$, of another is $u$, and of the rest is $0$. 
		In a run of the protocol starting from $X$, due to validity the network's decision value must be
		either $0$ or $u$ or $v$.
		We prove that none of these values can be decided in an equilibrium, reaching a contradiction.
	Consider some Agent $i$ and coalition $V \setminus \{i\}$.
	Define $s_v$ and $s_u$ as follows:
	\begin{itemize}
		\item $s_v :=$ a configuration in which $\#_{-i}(0) = n-2$, $\#_{-i}(v) = 1$
		\item $s_u :=$ a configuration in which $\#_{-i}(0) = n-2$, $\#_{-i}(u) = 1$
	\end{itemize}
	Assume towards a contradiction that 
	$P[D_i=0 |s_v \land I_i = u ] = p > 0$. Notice that $(s_v \land I_i = u) \in X$.
	
	By point $2$ of Lemma~\ref{lemma_multi_val_helper},
	if $I_i=v$ and the coalition acts as if their input vector is $s_v$, then $i$ must decide $0$.
	By Theorem~\ref{th_uniform}, $P[I_i=v] = \frac{1}{r}$, therefore,
	$P[D_i=0 | s_v] \geq \frac{1}{r} + \frac{p}{r}  > \frac{1}{r}$,
	contradicting Lemma~\ref{lemma_out_in}. Thus, in an equilibrium starting from configuration $X$, the decision value cannot be $0$.
	
	Assume towards a contradiction that: $P[D_i=v | s_v \land I_i = u ] = p > 0$.
	
	Notice that from point $1$ of Lemma~\ref{lemma_multi_val_helper},
	if $I_i=0$ and the coalition acts as if their input vector is $s_v$,
	then $i$ must decide upon $v$.
	As before we get:
	$P[D_i=v |s_v] \geq \frac{1}{r} + \frac{p}{r}  > \frac{1}{r}$,
	contradicting Lemma~\ref{lemma_out_in}. Thus, in an equilibrium starting from configuration $X$, the decision value cannot be $v$.
	
	Applying the symmetric claim for $u$, with a coalition that acts as if their input vector is $s_u$, we get that in an equilibrium starting from configuration $X$, the decision value cannot be $u$.
	
	Thus, no value from $\{0,u,v\}$ can be decided in
	an $(n-1)$-resilient equilibrium for multi-valued consensus starting with configuration $X$.
	Hence, due to validity there is no $(n-1)$-resilient equilibrium for $r$-valued consensus for any $r>2$. 
\end{proof}

\section{Necessary and Sufficient conditions for Deterministic Consensus}
\label{section_sufficient}

The necessary conditions from Section~\ref{section_necessary} are extended here into
necessary and \emph{sufficient} conditions for a \emph{deterministic} $(n-1)$-resilient equilibrium for binary consensus.
Deterministic means that the step of each agent in each round of the algorithm
is determined completely by its input and the history of messages it has received up until the current round.
In Appendix~\ref{section_non_det} some difficulties in trying to extend our proof to non-deterministic algorithms are provided. 
For the sufficient condition, a new problem -  Resilient Input Sharing (RIS),
a variant of knowledge sharing \cite{GTBlocks}, is introduced.
% \Xomit{The new version of Theorem~\ref{th_necessary} is now:}

\begin{theorem}%[Necessary and sufficient conditions for a deterministic $n-1$-strong equilibrium for consensus]
	\label{th_iff}
	A deterministic $(n-1)$-resilient equilibrium for consensus exists \emph{iff}:
	\begin{enumerate}
		\item $n$ is odd
		\item The input distribution is uniform
		\item There exists an algorithm for deterministic \EDP{} (defined below).
	\end{enumerate}
\end{theorem}

\subsection{The Resilient Input Sharing Problem}
\label{section_iks}
In the \EDP{} problem, agents in $V$ share their binary inputs while each agent $i$ assumes
$V \setminus \{i\}$ are in a coalition. Intuitively, each agent requires all other agents to commit their inputs before or simultaneously to them learning about its input. The motivation for this requirement is that we consider problems in which (1) all agents compute the same function on the inputs, and (2) if any one input is unknown, then any output in the range of the function is still equally possible  \cite{GTBlocks,GTColor}. Therefore the above requirement ensures that the coalition cannot affect the computation after learning the remaining (honest) agent's input, which is necessary for the computation to reach $(n-1)$-resilient equilibrium.
We use the following definitions:
\begin{itemize}
	
	\item $K_{j}^{t}$- Agent j's knowledge at the beginning of round $t$, including any information the coalition could have shared with it.   
	
	\item Agent $j$ is an $i$-knower($t$)- if at the beginning of round $t$ it can make a 'good' guess about $I_i$, i.e.,
	$ \exists b \in\{0,1\}: P[I_i = b |K_{j}^{t}] > P[I_i = b]$
	
	\item $Know(i,t)$ - the group of all $i$-knowers at the beginning of round $t$.
	In a \EDP{} algorithm, $Know(i,0) = \varnothing$ and $Know(i, \infty) = V \setminus \{i\}$
\end{itemize}

Consider for example the network in Figure~\ref{fig_know}.
At Round $0$, $A$ sends two different messages, whose XOR is its input, to $B$ and $C$.
At Round $1$, $B$ and $C$ can pass these messages to $D$,
even if this would not happen in a correct run. Thus: $Know(A,2) = \{D\}$, and $Know(A,3)=\{B,C,D\}$.

\begin{figure}
	\captionsetup{justification=centering}
	\includegraphics[scale=0.65]{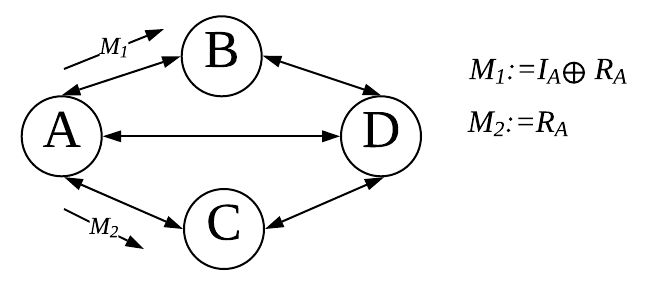}
	\centering
	\caption{Messages sent by agent $A$ at round $0$. $R_A$ is a random number chosen by $A$.}
	\label{fig_know}
\end{figure}

\subsubsection{The RIS Problem}
A solution to the \EDP{} problem satisfies the following conditions:
 
 \begin{enumerate}
 	\item \textbf{Termination} - the algorithm must eventually terminate.
 	\item \textbf{Input-sharing} - at termination, each agent knows the inputs of all other agents.
 	\item \textbf{Resilient} - at any round $t$, Agent $i$ does not receive new information from agents in $Know(i,t)$. % refrain from sending new information to Agent $i$.
 \end{enumerate}

\textbf{Notice}: in a consensus protocol, if $j$ is an $i$-knower($t$), and $j$ can still influence the output at round $t$, then the protocol is not an $(n-1)$-resilient equilibrium.
Thus, in an $(n-1)$-resilient equilibrium for consensus,
no new information can be sent to $i$ from any $i$-knower($t$) at round $t$.

\subsection{The effect of messages in a XOR computation}
\label{section_cons_2_IKS}
We prove that at the end of a distributed XOR computing algorithm, if an agent is given all the chains of messages that have affected its run, it can infer the input of every other agent (Theorem ~\ref{thm:inputEncoding}). This result applies for both deterministic and non-deterministic XOR algorithms.

\textbf{Remark $1$}: In synchronous networks, an agent can pass information to its neighbor through a silent round. Hereafter, every protocol in which informative silent rounds (explained in the proof of Lemma \ref{lemma:inputEncodingSize3} and defined formally in Appendix \ref{section_formal_meaningful}) occur is altered, and a special message \emph{EMPTY} is sent instead on the corresponding link.    

\textbf{Remark $2$}: Hereafter, we consider networks in which every agent knows the topology of the network before the algorithm starts. Otherwise, the coalition could always cheat and choose a topology in which \EDP{} is not possible (for example, 1-connected topology) 

\begin{definition}[Messages recipient]\label{def:RecipientOfMessages))}
	Let $R$ be a run of the protocol and $C \subseteq V$ a group of agents. $$ Recv(C,t,R)=\{i\in V | i
	\text{ received a message from } C \text{ in round } t \text{ of } R\} $$
\end{definition}
 
\begin{definition}[Agents affected by a message]\label{def:AgentsAffectedByMessage}
	In a run $R$, let $m$ be a message sent at round $t_{m}$ to $dst_m$ = agent $j$ from $src_m$. Then:
	\begin{itemize}
		\item $Aff_{(m,R,t_{m})}$ = $\{j\}$ - Agent $j$ is directly affected by $m$.
		\item $\forall k>0: Aff_{(m,R,t_{m}+k)}$ = $Aff_{(m,R,t_{m}+k-1)}$ $\cup$ Recv($Aff_{(m,R,t_{m}+k-1)}$,R,$t_{m}$+k) - Agents that were recursively affected by $m$.
		%\item $Aff_{(m,R)}$ = $Aff_{(m,R,T_{end})}$ ($T_{end}$ is the last round of $R$)
	\end{itemize}
\end{definition} 

$Aff_{(m,R,t)}$ illustrates that a message may affect more than just its recipient; Its potential effect propagates through the network, reaching different agents through other messages.

\begin{definition}[All the (chains of) messages that have an effect on agent $i$ in run $R$]\label{def:MessagesAffectingAgent))}
	\quad
	\begin{itemize}
		\item $Aff_{(i,R)}$ = $\{$ $<m,t_m,src_m,dst_m>$, $m$ sent in $R\ |i\in Aff_{(m,R,T_{end})}\}$ ($R$ terminates at $T_{end}$)
	\end{itemize}
\end{definition}

\begin{theorem}[The encoding of all inputs]\label{thm:inputEncoding}
	Let R be a run of a distributed XOR computing algorithm. Let $i,j\in V$,
	Agent $i$ can compute $I_j$ from the following information:
	\begin{enumerate} 
		\item $I_{i}$ - its input.
		\item Decision value i.e., the XOR of all inputs.
		\item $Aff_{(i,R)}$ - all the messages in $R$ that have an effect on Agent $i$.
	\end{enumerate}
\end{theorem}  
To prove Theorem \ref{thm:inputEncoding}, assume the following base case is correct (to be proved in the sequel):
\begin{lemma}\label{lemma:inputEncodingSize3}
	Theorem~\ref{thm:inputEncoding} is correct for a network of size $3$, $V=\{i,j,k\}$.
\end{lemma}

\begin{proof}[Proof of Theorem~\ref{thm:inputEncoding}]
	Let $G=(V,E)$ be a network where $n > 3$, such that $i,j \in V$.
	Create a new network $G'$ in which agents $i$ and $j$ are as in $G$, but all other agents in
	$V \setminus \{i,j\}$ are clustered into one 'virtual' agent $k$.
	A distributed XOR algorithm for $G'$ is:
	\begin{itemize}
		\item Agent $k$ chooses $n - 2$ bits such that the XOR of these bits is its $I_k$.
		\item Agents $i$ and $j$ behave in $G'$ as if they were in $G$, explicitly attaching to each message the id of its destination, while $k$ emulates the behavior of the other $n - 2$ agents in $V$, attaching to each message the id of its source.
	\end{itemize}
Let $I_i^R$ and $D_i^R$ be the input and output of $i$ in run $R$. For any run $R$ of the algorithm in $G$, $\exists R'$ - a run of the algorithm in $G'$ s.t.,:	\textbf{(1)} $I_i^R$=$I_i^{R'}$ , $I_j^R$=$I_j^{R'}$, \textbf{(2)} $D_i^R$=$D_i^{R'}$ and \textbf{(3)} $Aff_{(i,R)} \supseteq Aff_{(i,R')}$.\\
	From lemma~\ref{lemma:inputEncodingSize3} we know that from $D_i^{R'}$, $I_i^{R'}$ and $Aff_{(i,R')}$,  $I_j'$ can be computed. Therefore:
	$\forall i\neq j\in V$: - $D_i^{R}$,  $I_i^R$ and $Aff_{(i,R)}$ are enough to compute $I_j^R$.
\end{proof}
%It remains to prove lemma ~\ref{lemma:inputEncodingSize3}.
\begin{proof}[Proof of ~\ref{lemma:inputEncodingSize3}]
	$V = \{i,j,k\}$. Assume towards a contradiction that $\exists R_1, R_2$, two runs of the algorithm such that
	\begin{enumerate}
		\item $I_i^{R_1} $ = $ I_i^{R_2}$ - Agent $i$'s inputs in $R_1$ and $R_2$ are the same.
		\item $\bigoplus_{l \in V} I_l^{R_1}=\bigoplus_{l \in V} I_l^{R_2}$ - The decision value is the same in both $R_1$ and $R_2$.
		\item $Aff_{(i,R_1)}$ = $Aff_{(i,R_2)}$ - Exactly the same set of messages affect $i$ in both runs.
		\item $I_j^{R_1}\neq I_j^{R_2}$ - Agent $j$'s input in $R_1$ is different than in $R_2$.
	\end{enumerate}
	Clearly from 1, 2, and 4 it must be that $I_k^{R_1}\neq I_k^{R_2}$.

	Towards a contradiction we construct run $R_3$, in which  $i$'s and $k$'s inputs are the same as in $R_1$ and $j$'s input is the same as in $R_2$, but the decision value (XOR) in $R_3$ is the same as in $R_1$.
	
	In $R_3$, agents $i$ and $k$ start to perform their steps according to $R_1$ until the first round in which $i$ or $k$ receive a message that either does not receive in that round in $R_1$. Agent $j$ behaves the same as in $R_2$, until the first round, denoted round $T-1$, in which it receives a message $m$ it does not receive in that round in $R_2$. Notice that it is legal for all agents to act this way in round 0. Further, if $i$ and $k$ can continue according to $R_1$ and $j$ can continue according to $R_2$ until termination, then $i$ outputs the same value as it would in $R_1$, which is incorrect for $R_3$.
	
	\begin{description}
		\item[Observation $1$] From round $T$ until termination
		$j$ cannot send messages to $i$ in either $R_1$ or $R_2$
		or otherwise, $m$'s effect would propagate to $i$, causing - $Aff_{(i,R_1)}\neq Aff_{(i,R_2)}$, contradicting point $3$ of the assumptions.
		\item[Observation $2$] Similarly from round $T$ until termination, $j$ cannot send messages to $i$ in $R_3$ or otherwise, let
		$t\geq T$ be the first round (after $T$) of $R_3$
		in which $j$ sends a message to $i$.
		In $R_1$ - $j$ does not send a message to $i$ in round $t$ (see Observation $1$).
		This means that this silent round $t$ of $R_1$ between $j$ and $i$ is informative
		(it tells $i$ that the run is $R_1/R_2$ and not $R_3$).
		Since we do not allow informative silent rounds (see Remark $1$), we reach a contradiction.
	\end{description}
	Notice that by point $3$ in the assumptions, after $T$ $j$ cannot even communicate with $i$ through $k$,
	since $m$'s effect would propagate to $i$ through $k$.
	From the two observations above, from round $T$ of $R_3$,
	$j$ cannot communicate with $i$, and from $i$'s perspective, $j$ is running $R_1$.
	The same logic applies for $k$ - the first round in which it is illegal for $k$ to act according to $R_1$,
	is a round after which $k$ cannot send messages to $i$ (even not through $j$).
	Thus $i$'s experience throughout $R_3$ is the same as in $R_1$,
	resulting in $i$ making an incorrect output.  Contradiction.	  
\end{proof}

\subsection{Deterministic $(n-1)$-resilient Consensus implies \EDP{}, completing the proof}
\label{section_obs_det_cons}
 In a deterministic synchronous binary consensus protocol, in which all agents start at the same round, for each input vector the run of the algorithm is fully determined.

Let us look at a network running some deterministic binary consensus,
with agent $i\in V$ and coalition $V \setminus \{i\}$.
Intuitively, agents in the coalition can choose in advance an input vector to be used in the algorithm. Thus, from the coalition's perspective, there can be only two possible runs - $R_0$ in which $I_i=0$, and $R_1$ in which $I_i=1$. For each agent in the coalition, there is the first round in which $R_0$ and $R_1$ differ, at that point this agent knows $I_i$. Thus, each agent in the coalition is in one of two states - knows nothing about $I_i$ or knows $I_i$, this is in contrast to non-deterministic algorithms, see for example Figure~\ref{fig_know}.

Below we transform any deterministic $(n-1)$-resilient equilibrium for binary consensus into a deterministic \EDP{}. In Appendix ~\ref{section_non_det} the difficulties in the non-deterministic case are explained.  
\begin{theorem}
	\label{thm:consImpliesRKS}
	If there exists a deterministic $(n-1)$-resilient equilibrium for binary consensus,
	$A$ on network $G=(V,E)$ then there exists an algorithm $\tilde{A}$ for \EDP{}, on $G$.
\end{theorem}

\begin{proof}
	In $\tilde{A}$, each agent $i$ runs $A$ with the following modifications:
	\begin{itemize}
		\item For each message $m$ that $i$ receives, $i$ appends $<m,src_{m},dst_{m},t_{m}>$ to a local buffer $B$ of messages that has affected it.
		\item Agent $i$ appends $B$ to each message it sends.
		\item Agent $i$ adds to $B$ all the information piggy-bagged on incoming messages.
	\end{itemize}
	In this new algorithm $\tilde{A}$, every message propagates in the network, reaching all the agents it affects. By the end of the algorithm, the buffer maintained by agent $i$ contains $Aff_{(i,R)}$, where $R$ is the run of the original consensus protocol $A$. By theorem~\ref{th_xor}, $A$ is a XOR computing protocol, and by theorem~\ref{thm:inputEncoding}, $i$'s buffer contains enough information to infer all inputs. Thus $\tilde{A}$ is an \EDP{} protocol.
	
	It remains to prove that $\tilde{A}$ is resilient. An input sharing protocol is resilient (Subsection~\ref{section_iks}) if at any round $t$, $i$ does not receive new information from agents in $Know(i,t)$. As stated before, this demand applies for $(n-1)$-resilient equilibrium for binary consensus as well.
	Thus, to show that $\tilde{A}$ is resilient, it is enough to show that $\forall i\in V$:
	\begin{itemize}
		\item In each round $t$ of $\tilde{A}$, $i$ receives messages from the same neighbors it receives from in $A$
		\item In each round $t$ of $\tilde{A}$, $\forall j\neq i$: $j\in Know(i,t)$ in $\tilde{A}$ $\implies$ $j\in Know(i,t)$ in $A$  
	\end{itemize}
	The first point is immediate from the construction of $\tilde{A}$.
	For the second point - observe some agent $j$ at round $t$ of $A$, which is not an $i$-knower in $A$. For $j$ to become an $i$-knower($t$) in $\tilde{A}$, the coalition must send $j$ enough information by $t$ for it to make a 'good' guess about $I_i$.
	There are two kind of paths in $G$ by which the coalition can send information to $j$ - paths that do not pass through $i$, and paths that do.
	
	Through paths not including $i$, the coalition can pass information in the same pace for both $A$ and $\tilde{A}$. Since $j\notin Know(i,t)$ in $A$, using these paths alone is not enough to make $j$ an $i$-knower($t$) in $\tilde{A}$.
	Regarding paths that include $i$ - as argued in the beginning of this subsection, in a \emph{deterministic} $(n-1)$-resilient equilibrium for binary consensus, if a member of the coalition has any information about $I_i$, then that member \emph{knows} $I_i$. Therefor, in $A$, $i$ should not receive messages from members of $Know(i,t)$ at round $t$. Thus if the coalition has information it wants to pass to $j$, it cannot do so using paths including agent $i$, since $i$ does not accept and propagate messages from $i$-knowers.
	To conclude, if $j$ is an $i$-knower in $\tilde{A}$, $j$ is an $i$-knower in $A$. Since $A$ is $(n-1)$-resilient equilibrium for consensus, $\tilde{A}$ is resilient as well.   
\end{proof}

%We can finally prove the necessary and sufficient conditions from Theorem~\ref{th_iff}: 

\subsubsection{Completing the proof, necessary and sufficient conditions for deterministic Consensus}
\label{section_proof}

\begin{proof}[{\bf Proof of Theorem~\ref{th_iff}$\Leftarrow$}]
	Assume that the $3$ conditions are realized, and let us suggest a simple
	$(n-1)$-resilient equilibrium for binary consensus:
	run the \EDP{} algorithm and output the XOR of all inputs.
	Since the \EDP{} algorithm is resilient, no coalition has an incentive to cheat.
\end{proof}

\begin{proof}[{\bf Proof of Theorem~\ref{th_iff}$\Rightarrow$}]
	Assume that $(n-1)$-resilient equilibrium for binary consensus exists.
	By \ref{cor_noeven} and \ref{th_uniform},
    $n$ is odd and the input distribution is uniform.
	By theorem~\ref{thm:consImpliesRKS}, \EDP{} is possible.
\end{proof}

\section{Discussion}
\label{section_discussion}

Surprisingly, while there is an equilibrium for binary consensus resilient to coalitions of $n-1$ agents, no such equilibrium exists for multi valued consensus.
This is the first model we know of in which there is
a separation between binary and multi valued consensus.
Intuitively, this is because a coalition with a preference towards $v$ has an incentive to cheat
and act as if the input of all agents in the coalition is $v$,
thus lowering the number of possible decision values (due to validity) to two values, at most.
Consider for example the standard bit-by-bit reduction from binary to multi valued consensus,
the probability to decide $v$ is now at least $\frac{1}{2}$
instead of $\frac{1}{r}$, since the decision value is determined by the
decision on the first bit of the coalition input that differs from the input of the honest agent.
We conjecture that this intuition holds even for smaller coalitions, up to a single cheater. The results in $\S \ref{section_necessary}$ and $\S \ref{section_multi}$ hold regardless of the network topology, scheduling models,
or cryptographic solutions,
as they are based solely on the input values  and utility of the agents.

Furthermore, we present necessary and sufficient conditions for
$(n-1)$-resilient equilibrium for binary \emph{deterministic} consensus using the
resilient input sharing (RIS) problem.
This in fact means that an agent cannot hide its input from the rest of the network
in any $(n-1)$-resilient equilibrium protocol that computes XOR,
%\fix{ even if the agents of the rest of the network does not fully cooperate, and are willing to share with each other only messages sent throughout the execution (and not their full inputs).}
i.e., even though we only compute the XOR of inputs,
at the end of the protocol all agents can deduce the input values of all other agents.

There are several open directions for research:
\begin{itemize}
	\item Extending the equivalence result to \emph{non-deterministic} consensus and RIS.
	
	\item Can binary consensus be solved without the conditions of even size and uniform input
	for coalitions of a smaller size, such as $n-2$ or $\frac{n}{2}$?

	\item Does an equilibrium for multi-valued consensus exist for coalitions of size $n-2$ or less?
\end{itemize}

%\bibliography{gt_cons}

\clearpage
\appendix

\section{$(n-2)$-resilient Consensus for Even $n$}
\label{section_res_even}
%\fix{
In \cite{GTBlocks} the authors provide a different protocol for even and odd size networks.
Here we prove that the protocol suggested for binary consensus when $n$ is even,
provides $(n-2)$-resilient equilibrium for binary consensus.
The protocol assumes the existence of an $(n-2)$-resilient equilibrium for knowledge-sharing in order to perform an $(n-2)$-resilient equilibrium for leader election (notice that in \cite{YifrachLeader}, it is shown that in an asynchronous ring, the leader election algorithm of \cite{GTBlocks} is not resilient to coalitions of size $O(\sqrt{n})$).
Further, the protocol assumes each agent has a unique id, and all agents start the protocol at the same round.

\begin{algorithm}
	\caption{\cite{GTBlocks} Protocol for $(n-2)$-resilient equilibrium for binary consensus}
	\label{proto_even}
	for each agent $i$:
	\begin{algorithmic}
		\State 1. Let $r_i=random(1,...,n)$
		\State 2. Execute Knowledge sharing \cite{GTBlocks} to learn $K = \{<I_1,r_1,id_1>,...<I_n,r_n,id_n>\}$
		\State 3. For any $k$, if $I_k \notin \{0,1\}$ or $r_k \notin \{1,...,n\}$, or $\exists j$ such that $id_j=id_k$,  set $D_i=\perp$ and terminate
		\State 4. Calculate $L = (\sum_{k=1}^{n} r_k) \mod n$, set $leader$ to be the L-th ranked id 
		\State 5. Set $D_i = \bigoplus\limits_{\substack{k \in \{1,...,n\}\\ id_{k}\neq leader}} I_k$ and terminate
	\end{algorithmic}
\end{algorithm}

Essentially, the protocol suggested in \cite{GTBlocks} when $n$ is even, performs input sharing in parallel to leader election, then outputs the XOR of all inputs without the leader's input.
It is easy to see that this protocol for consensus satisfies agreement, validity, and termination.

For the rest of this section, let $V$ be the agents in a network of even size,
executing the protocol in Algorithm~\ref{proto_even}, with coalition $C=V\setminus\{i,j\}$.
Also, we assume the input distribution is uniform, i.e., $\forall i\in V: P[I_i = 0] = P[I_i=1] = \frac{1}{2}$.

\begin{theorem}
	\label{th_equil_even}
	Algorithm \ref{proto_even} is an $(n-2)$-resilient-equilibrium for binary consensus.
\end{theorem}

\noindent Proof of Theorem~\ref{th_equil_even} follows after the following observation and lemmas:

If no agent deviates from the protocol in Algorithm~\ref{proto_even}, then the decision value of the consensus is uniformly distributed. Therefore, if after the coalition $C$ deviates, the probability to decide a value preferred by $C$ is still $\leq 1/2$,  $C$ has no incentive to cheat.

\begin{lemma}
	\label{2vsC}
	If at the end of the knowledge sharing step, $i$ learns the true value of $I_j$ (or vice versa), $C$ has no incentive to deviated from the protocol.
\end{lemma}
\begin{proof}
In this case at least one of the inputs of agents $i$  and $j$ is not omitted from the XOR performed by $i$ or/and $j$. Since the coalition has no influence on these inputs, which are uniformly distributed, the result of the XOR is also uniformly distributed, and they have no incentive to cheat.
\end{proof}

Following Observation A.0 and Lemma \ref{2vsC} it remains to consider the case in which the coalition can cheat each of $i$ and $j$ about the input, or id, and/or random value selected in step $1$, of the other.
\begin{lemma}
	\label{diff_leader}
	$C$ has no incentive to share with $i$ and $j$ two sets of ids and random values that disagree.
\end{lemma}
\begin{proof}
	Assume $C$ has a preference towards $v$.
	Denote by $X$ the case in which the coalition forced $i$ and $j$ to elect two different leaders. Notice that to achieve this the coalition must provide $i$ and $j$ two different sets of ids and random values for all the other agents.
		
	In case $X$  the decision value of $i$ is independent of the decision value of $j$. 
	Following \cite{GTBlocks} $\forall k,l\in V: P[leader_k = l] = \frac{1}{n}$. Thus,
	$$\forall k\in\{i,j\}:\ P[leader_k\neq k] = \frac{n-1}{n}$$
	If $i$ does not elect itself as leader, then (based of the uniform input distribution) $D_i=v$ with probability $\frac{1}{2}$.
	Hence:
	
	\begin{align}
	P[D_i = v | X] &= P[leader_i\neq i | X]\cdot \frac{1}{2} + P[leader_i= i | X]\cdot P[D_i= v | X \land leader_i= i ]  \notag\\
	&\leq P[leader_i\neq i | X]\cdot \frac{1}{2} + \frac{1}{n} \notag\\
	&= \frac{n-1}{2n} + \frac{1}{n} \notag\\
	&= \frac{n+1}{2n} \notag
	\end{align}
	The same goes for agent $j$.
	Since the decision of $i$ is independent of the decision of $j$ (by solution preference, the coalition succeeds only if $D_i=D_j=v$):  
	$$\forall n>2: P[D_i=v \land D_j=v | X] = P[D_i=v | X] \cdot P[D_j=v | X] \leq {(\frac{n+1}{2n})}^2 < \frac{1}{2}$$
	Since the probability to decide $v$ when executing the protocol in Algorithm~\ref{proto_even}
	with no deviation is $\frac{1}{2}$,
	there is no incentive for $C$ to share different ids or random values with $i$
	than it shares with $j$ (and vice versa).
\end{proof}
\begin{lemma}
	\label{diff_input}
	$C$ has no incentive to share a set of input values with $i$ and a set with $j$, that disagree. 
\end{lemma}
\begin{proof}
	Assume $C$ has a preference towards $v$, and denote by $Y$ the case in which the coalition provides a set of input values with $i$ and a set with $j$, that disagree.  Like in the previous proof, the decision values of $i$ and $j$ are independent.
	By~\ref{diff_leader}, both agents $i$ and $j$ elect the same leader, hence that at least one of them is not elected. \emph{W.l.o.g}, $i$ is not the leader. When $i$ calculates the XOR (step $5$), $I_i$ is not omitted from the calculation. Since the set of inputs provided to $i$ is independent for $I_i$ (provided by knowledge-sharing being resilient), and since $C$ does not know in advance $I_i$, which is uniformly distributed, the result of the XOR is uniformly distributed. I.e.: $P[D_i = v | Y] = \frac{1}{2}$.
	Since the probability to reach consensus on $v$ when running Algorithm \ref{proto_even} with no deviation is $\frac{1}{2}$, there is no incentive for $C$ to share different input values with $i$, than it shares with $j$.
\end{proof}

\begin{proof}[Proof of Theorem~\ref{th_equil_even}]
	From lemmas ~\ref{2vsC}, ~\ref{diff_leader} and~\ref{diff_input}, we know that, in any run of the algorithm, both $i$ and $j$ obtain the same knowledge $K$. Since the decision value is uniformly distributed in a correct run, then for any legal knowledge sharing $K$: $P[D_i=0] = P[D_i=1] = \frac{1}{2}$.
	This means that $C$ has no incentive to choose in advance either a specific set of random values or input values or ids.
\end{proof}
%}

\section{Informative Silent Rounds and Informative Messages}
\label{section_formal_meaningful}
For this section, let $R$ be a run of a distributed XOR algorithm $A$ in network $G=(V,E)$.
\begin{definition}[Link experiences]\label{def:LinkExp}
	For any Agent $i\in V$ at any round $t$, for all $(i,j)\in E$ define the incoming link experience of $i$ to be:
	$$
	ILE(i,j,t) = 
	\begin{cases}
	\text{m} &\quad \text{ (j sends message m to i at round t)} \\
	\text{silence} &\quad \text{ (j does not send any message to i at round t)}
	\end{cases}
	$$
	Similarly, define the outgoing link experience of $i$ with $j$ at round $t$ to be:
	$$
	OLE(i,j,t) = 
	\begin{cases}
	\text{m} &\quad \text{ (i sends message m to j at round t)} \\
	\text{silence} &\quad \text{ (i does not send any message to j at round t)}
	\end{cases}
	$$
\end{definition}

\begin{definition}[Round of an agent]\label{def:RoundOfAnAgent}
	For $i \in V$ at round $t$:
	\begin{itemize} 
		\item $I_i$     := Agent $i$'s input.
		\item $in(i,t)$ := All incoming link experiences $i$ has with its neighbors at round $t$.
		\item $out(i,t)$ := All outgoing link experiences $i$ has with its neighbors at round $t$. 
		\item $D(i,t) \in \{0,1,?\}$ := The decision value of $i$. As long as $t$ is not the final round, $D(i,t)=?$ 
		\item $round(i,t) = <I_i,in(i,t),out(i,t),D(i,t)>$ := Round $t$ from agent $i$'s perspective
	\end{itemize}
\end{definition}

\begin{definition}[Run of an agent]\label{def:RunOfAnAgent}
	For $i \in V$, define $R(i)$ to be the projection of $R$ on $i$:
	$$R(i) := <round(i,0),round(i,1),... round(i,T_{end})>$$
\end{definition}

\begin{definition}[Prefix and suffix of a run]
	For $i \in V$:
	$$R(i)^{0...t} := <round(i,0),round(i,1),... round(i,t)>$$
	I.e. the prefix of $R(i)$ up to round $t$.
	Each prefix of a run has a set of possible legal suffixes of the form:
	$$S(i)^{t+1...} := <round(i,t+1),round(i,t+2),...>$$
\end{definition}

\begin{definition}[Informative link experience]
	Intuitively, informative link experiences are $ILE$ after which $i$'s execution may be altered. 
	Let $i,j\in V$. Denote $e_1$  to be a legal $ILE$ that $i$ has at round $t$ of $R$ with $j$. $e_1$ is informative if there exists:
	\begin{itemize}
		\item $e_2$ := Another $ILE$ that $i$ has with $j$ at round $t$ ($e_1\neq e_2$) 
		\item $out$ := A set of outgoing link experiences $i$ had with its neighbors at round $t$
		\item $in$  := A set of incoming link experiences $i$ had with its neighbors at round $t$ \emph{not including $j$}.
		\item $D$   := A decision value.
	\end{itemize}
	Such that the following holds:
	\begin{enumerate}
		\item Both $<I_i,\ in \bigcup \{e_1\},\ out,\ D>$ and $<I_i,\ in \bigcup \{e_2\},\ out,\ D>$ are legal rounds for agent $i$ in a run of $A$ with prefix $R(i)^{0...t-1}$
		\item $\exists$ $S(i)^{t+1...}$ - a suffix of $i$'s run, such that:
		\begin{equation*}
		\begin{aligned}
			P[S(i)^{t+1...}\ |\ R(i)^{0...t-1} \land <I_i,\ in \bigcup \{e_1\},\ out,\ D>] \neq \\ P[S(i)^{t+1...}\ |\ R(i)^{0...t-1} \land <I_i,\ in \bigcup \{e_2\},\ out,\ D>]
		\end{aligned}
		\end{equation*}
	\end{enumerate} 
\end{definition}

\begin{definition}[Informative silent round]
	In Subsection~\ref{section_cons_2_IKS}, an informative silent round is actually an incoming link experience $i$ has with $j$ at round $t$, such that:
	\begin{enumerate}
		\item $ILE(i,j,t) =$ silence
		\item $ILE(i,j,t)$ is informative
	\end{enumerate}
\end{definition}

\section{Difficulties in extending Theorem \ref{thm:consImpliesRKS} to Non Deterministic Case}
\label{section_non_det}
\begin{figure}[H]
	\captionsetup{justification=centering}
	\includegraphics[scale=0.70]{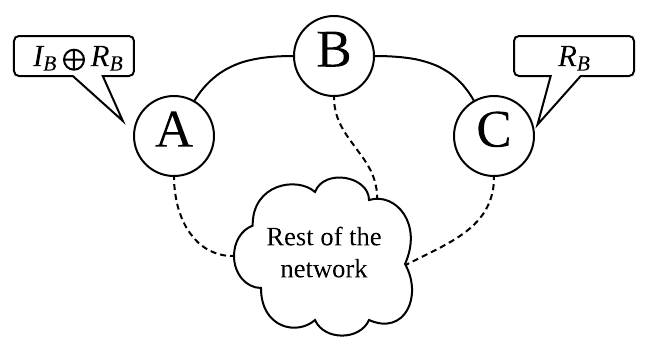}
	\centering
	\caption{A snippet of agents $A$ and $C$'s knowledge regarding $I_B$
		in a non-deterministic XOR computing algorithm. $R_B$ is a random number chosen by $B$.}
	\label{fig_non_det_problem}
\end{figure}
Figure~\ref{fig_non_det_problem} depicts a counter example in a non-deterministic algorithm to the construction in Theorem~\ref{thm:consImpliesRKS}. $A$ and $C$ cannot make a good guess regarding $I_B$ on their own. If however, they were able to combine the information they have acquired, they would become $B$-knowers.
In the original algorithm, $B$ can still receive (send) messages from (to) $A$ and $C$ (they are not $B$-knowers).
Applying the construction in Theorem~\ref{thm:consImpliesRKS} on this non-deterministic algorithm, agent $A$ would have been able to pass $C$ its array of messages, and $B$ would have to let it pass through, thus creating an $A$-$C$ 'shortcut' through $B$.

\end{document}